\documentclass[11pt]{article}
\usepackage[latin9]{inputenc}
\usepackage[letterpaper]{geometry}
\usepackage{xcolor}
\usepackage{amsmath}
\usepackage{amssymb}
\usepackage{esint}
\PassOptionsToPackage{normalem}{ulem}
\usepackage{ulem}

\makeatletter

\providecolor{lyxadded}{rgb}{0,0,1}
\providecolor{lyxdeleted}{rgb}{1,0,0}

\usepackage{amsthm}
\usepackage{amsfonts}
\usepackage{amsthm}
\usepackage{mathrsfs}
\usepackage{subfigure}

\numberwithin{equation}{section}

\let\textquotedbl="
\newtheorem{theorem}{Theorem}

\newtheorem{lm}[theorem]{Lemma}

\newtheorem{prp}{Proposition}

\title{Slowly changing potential problems in Quantum Mechanics: Adiabatic Theorems, Ergodic Theorems, and Scattering}
\author{S. Fishman and A. Soffer}
\numberwithin{equation}{section}

\makeatother

\begin{document}
\maketitle

\noindent \textbf{Abstract}

We employ the recently developed multi-time scale averaging method
to study the large time behavior of slowly changing (in time) Hamiltonians.
We treat some known cases in a new way, such as the Zener problem,
and we give another proof of the Adiabatic Theorem in the gapless
case. We prove a new Uniform Ergodic Theorem for slowly changing unitary
operators. This theorem is then used to derive the adiabatic theorem,
do the scattering theory for such Hamiltonians, and prove some classical
propagation estimates and Asymptotic Completeness.

\section{Introduction}

The dynamics of a system perturbed by a time-dependent potential is
difficult to analyze in general terms. This problem shows up in many
fields as adiabatic processes in Quantum Mechanics, as effective theories
in scattering and nonlinear dynamics, in noisy and/or periodically
perturbed systems, and more.

A ``special'' case of fundamental importance is the possibility
of identifying a small parameter that controls the rate of change
of the system/perturbation. The most celebrated example is the \textit{adiabatic
theorem} in QM e.g., \cite{HJ,ASY}, where one assumes that the perturbing
potential, $W$, changes with time as 
\begin{equation}
W(x,\varepsilon t)
\end{equation}
where $\epsilon$ is assumed small, \textit{and} the time dependence
vanishes after some finite time, that is $W\left(x,\epsilon t)\right)=W\left(x,\infty)\right)$
for $\forall t>t_{\infty}$ typically $t_{\infty}$ is of order $1/\varepsilon$.
However, many other examples exist where the slow change never really
stops, or continues for a time much larger than $1/\varepsilon$.
In this case, the standard adiabatic theorem fails \cite{Ber1,Ber2}

In this work we will study three different examples. The Landau-Zener(LZ)
type Hamiltonian, the gapless adiabatic theorem in Quantum Mechanics
and scattering theory for some slowly changing potentials.

The LZ model is discussed in section 3. We study in detail the behavior
in time for small, medium and large times, using our new multi-scale
time averaging method \cite{FS}. The method we use is of general
nature, can be applied in a similar way to the general form of such
Hamiltonian.

In particular we show the time scales on which the LZ Hamiltonian
has a nontrivial action, with error estimates. Then, in section 4
we consider the adiabatic theorem in Quantum mechanics. First we prove
a uniform adiabatic ergodic theorem. That is we show that for Hamiltonians
which change slowly in time, the RAGE theorem holds. See theorem 4. 

Then it is applied to give a new proof of the gapless adiabatic theorem.

In section 5 we consider the scattering problem on a short-range slowly
changing(i.e adiabatic) potentials.

We prove some basic propagation estimates and then for a class of
potentials, we show that the limits defining the scattering matrix
are uniform in $\epsilon\rightarrow0$ provided the potential perturbation
$V\left(x,\epsilon t\right)=V\left(x,\infty\right)$ for $t\geq\text{c}exp\left(\epsilon^{-1/4}\right)$

An important example is the Landau-Zener type Hamiltonians, the simplest
of which is the $2\times2$ case given by

\begin{eqnarray}
H(t)=\left(\begin{array}{ll}
\varepsilon t & B\\
B & -\varepsilon t
\end{array}\right)\ \text{acting on }\mathbf{C}^{2}.
\end{eqnarray}

In this work, we study the problem of slowly changing potentials,
with or without the assumption of finite ``life-time''. We consider
in detail the LZ model: our aim is to prove that it is asymptotically
stable/complete, in the sense that the asymptotic Hamiltonian is an
explicit, time-independent operator. Moreover, we find the \textit{interaction
time} \cite{Vit1,Vit2,FMB} that is, the time after which the dynamics
is time-independent up to \textit{explicitly given} small correction.

We also analyze the short time behavior of the system. It turns out
to be oscillatory. We then analyze other problems that can be treated
by the same approach: scattering theory and the adiabatic case. For
this, we prove a new Uniform Ergodic Theorem.

The method we employ is a multiscale time averaging technique, recently
developed by us \cite{FS}. In this approach, we study by time-averaging,
the dynamics given by 

\begin{equation}
i\frac{\partial\psi}{\partial t}=\beta A(t)\psi
\end{equation}
where $|\beta|\ll1$, $A(t)$ is a family of self-adjoint operators
acting in a Hilbert space ${\cal H}$ with initial condition $\psi(t=0)=\psi\in{\cal H}$.

We show that by averaging over a time interval of order $\frac{1}{\sqrt{\beta}}$,
followed by a normal form transformation, the above problem is replaced
by a new, piecewise time-independent problem generated by the averaged
Hamiltonian, and a correction generated by a system as above, but
with $A(t)\rightarrow\tilde{A}(t)$ and $\beta\rightarrow C\beta^{\frac{3}{2}}$.
This process can be redone repeatedly, to all orders in $\beta$.

\section{Slowly changing potentials and time averaging\label{sec:Slowly-changing-potentials}}


Consider the problem 
\begin{align}
i\frac{\partial\psi}{\partial t} & =\left(-\Delta+W(\varepsilon t,x)\right)\psi\label{eq:1.1}\\
 & \psi(t=0)=\psi_{0}\in L_{2}\left(\mathbb{R}^{n}\right),\nonumber 
\end{align}
where $W(\varepsilon t,x)$ is a smooth bounded function such that
for some $\sigma>\sigma_{0}>1$, $0<\varepsilon<<1$ 
\begin{align}
\sup_{t}\left|\left|\langle x\rangle^{\sigma}W(\varepsilon t,x)\right|\right|_{\mathrm{L}^{\infty}}<C_{0}<\infty\label{eq:1.2}
\end{align}
Here 
\begin{equation}
\langle x\rangle^{2}=1+|x|^{2}.
\end{equation}

We are interested in the behavior of the solutions of such an equation,
for $t>1$. We assume 
\begin{equation}
\sup_{\tau}\left|\left|\langle x\rangle^{\sigma}\frac{\partial W(\tau,x)}{\partial\tau}\right|\right|_{\mathrm{L}^{\infty}}<C_{1}<\infty.\label{eq:1.3}
\end{equation}
The solution exists for all times, in $\mathrm{L}^{2}$. See Reed-Simon
\cite{RSI}. Traditionally, this problem is treated by the adiabatic
theory. Here, we will show how to treat this and similar problems
by multiscale time averaging techniques. We will use the multiscale
time-averaging of \cite{FS}. Let $T_{0}=1/\sqrt{\varepsilon}$.

Define 
\begin{equation}
\overline{W}(x)=\frac{1}{T_{0}}\int_{0}^{T_{0}}W(\varepsilon t,x)\ \mathrm{d}t=\frac{1}{\varepsilon T_{0}}\int_{0}^{\varepsilon T_{0}}W(\tau,x)\ \mathrm{d}\tau
\end{equation}
Then 
\begin{equation}
\begin{aligned}W(\varepsilon t & ,x)-\overline{W}(x)=\frac{-1}{\varepsilon T_{0}}\int_{0}^{\varepsilon T_{0}}\left[W(\tau,x)-W(\varepsilon t,x)\right]\mathrm{d}\tau\\
 & =\frac{-1}{\varepsilon T_{0}}\int_{0}^{\varepsilon T_{0}}W^{\prime}(y(\tau),x)(\tau-\varepsilon t)d\tau\\
\ \ y=y(\tau),
\end{aligned}
\label{eq:1.4}
\end{equation}
where we used that $\left.W(\tau,x)-W(t,x)=\frac{\partial W}{\partial y^{\prime}}(y^{\prime},x)\right|_{y^{\prime}=y(\tau)}\left(\tau-t\right)$
for some $y(\tau)\in\left[\tau,t\right]$. For $t\leq T_{0}$, one
has $\left|\tau-\varepsilon t\right|\leq\varepsilon T_{0}$, and therefore,
we have: 

\noindent \begin{lm}
\begin{equation}
\sup_{t\leq T_{0}}\sup_{x}\left|W(\varepsilon t,x)-\overline{W}(x)\right|\langle x\rangle^{\sigma}\leq C_{1}\varepsilon T_{0}=C_{1}\sqrt{\varepsilon}.\label{eq:1.5}
\end{equation}
Hence, if we let (see \cite{FS}) 
\begin{equation}
\overline{W}_{j}(x)=\frac{1}{T_{0}}\int_{jT_{0}}^{(j+1)T_{0}}W(\varepsilon t,x)\ \mathrm{d}t,\label{eq:1.6}
\end{equation}
then for $jT_{0}\leq t\leq(j+1)T_{0}$ 
\begin{equation}
H(\varepsilon t)-\left(-\Delta+\overline{W}_{j}(x)\right)=\mathcal{O}\left(\langle x\rangle^{-\sigma}\right)\sqrt{\varepsilon}.\label{eq:1.7}
\end{equation}

\end{lm}

Let 
\begin{equation}
\widetilde{A}(t)=\left(H(\varepsilon t)-H_{0}^{(g)}(\varepsilon t)\right)\varepsilon^{-1/2}\label{eq:1.8}
\end{equation}
where 
\begin{equation}
H_{0}^{(g)}(\varepsilon t)=-\Delta+\overline{W}_{j}(x)\text{ for }jT_{0}\leq t\leq(j+1)T_{0}.\label{eq:1.9}
\end{equation}
Then, by~\eqref{eq:1.7}, $||\langle x\rangle^{\sigma}\tilde{A}(t)||_{L^{2}\cap L^{\infty}}\equiv\left|\left|\langle x\rangle^{\sigma}\tilde{A}(t)\right|\right|_{L^{2}}+\left|\left|\langle x\rangle^{\sigma}\tilde{A}(t)\right|\right|_{L^{\infty}}=C(C_{0}+C_{1})\leq O(1)$,
where $C$ is a universal constant independent of the Hamiltonian.
Therefore, we can rewrite (\ref{eq:1.1}) in the form 
\begin{equation}
i\frac{\partial\psi}{\partial t}=\left[\varepsilon^{1/2}\widetilde{A}(t)+H_{0}^{(g)}(\varepsilon t)\right]\psi(t).\label{eq:1.10}
\end{equation}
This implies that \begin{subequations} 
\begin{align}
\psi(t) & =V(t)\widetilde{\psi}(t)\label{eq:1.11b}
\end{align}
with time ordering (see Eq. 2.3 \cite{FS} )
\begin{equation}
V(t)=\mathcal{T}e^{-i\int_{0}^{t}H_{0}^{(g)}(\varepsilon s)\ \mathrm{d}{s}}\label{eq:1.11bb}
\end{equation}
\end{subequations} and 
\begin{equation}
\left(\varepsilon^{1/2}\widetilde{A}(t)+H_{0}^{(g)}\right)\psi=i\frac{\partial\psi}{\partial t}=H_{0}^{(g)}(\varepsilon t)\psi(t)+V(t)i\frac{\partial\widetilde{\psi}}{\partial t}.
\end{equation}
The second equality is found by direct differentiation of (\ref{eq:1.11b})
Therefore 
\begin{equation}
V(t)i\frac{\partial\widetilde{\psi}}{\partial t}=\varepsilon^{1/2}\widetilde{A}(t)V(t)\widetilde{\psi}(t)
\end{equation}
or 
\begin{equation}
i\frac{\partial\widetilde{\psi}}{\partial t}=\varepsilon^{1/2}V(t)^{-1}\widetilde{A}(t)V(t)\widetilde{\psi}(t)\equiv\beta A(t)\widetilde{\psi}(t).\label{eq:1.12}
\end{equation}
We have used \ref{eq:1.11bb} which implies 
\begin{align}
\psi(t)=\mathcal{T}e^{-i\int_{0}^{t}H^{(g)}(\varepsilon s)\ \mathrm{d}s}\widetilde{\psi}(t)
\end{align}
and then by 
\begin{equation}
V(t)^{-1}\tilde{A}(t)V(t)\equiv A(t);\qquad\beta\equiv\sqrt{\varepsilon}.\label{eq:1.13}
\end{equation}
We now apply the multiscale-averaging of \cite{FS} to Equation~\eqref{eq:1.12}
where $A(t)$ plays the role of the Hamiltonian. By Theorem 3 there:
\begin{subequations} 
\begin{align}
\widetilde{\psi}(t) & \equiv U_{0}(t)U_{1}(t)\widetilde{U}_{2}^{-1}\phi_{2}^{U}(t)\label{eq:1.14a}\\
 & i\frac{\partial}{\partial t}\phi_{2}^{U}(t)=\beta^{3/2}A_{\mathrm{NF}}(t)\phi_{2}^{U}(t)\label{eq:1.14b}\\
 & \left|\left|A_{NF}(t)\right|\right|=\mathcal{O}(1).\label{eq:1.14c}
\end{align}
\end{subequations} 

Where $A_{NF}$ is obtained by a normal for transformation of the
potential averaging,$A_{NF}$ is denoted by $\tilde{A}$ and it is
given by \eqref{eq:1.13} and $\phi_{2}^{U}=\psi\left(0\right)$

Here $U_{0}(t)$ is generated by 
\begin{equation}
\begin{aligned}i\partial_{t} & U_{0}(t)=\beta\overline{A}_{0}^{g}(t)U_{0}(t)\\
\\
\\
\end{aligned}
\label{eq:1.15}
\end{equation}
where 
\begin{equation}
\overline{A}_{0}^{g}(t)\equiv\overline{A}_{0}^{(n)},\qquad\mathrm{for}\ nT_{0}\leq t<(n+1)T_{0}
\end{equation}
and 
\begin{equation}
\overline{A}_{0}^{(n)}\equiv\frac{1}{T_{0}}\int_{nT_{0}}^{(n+1)T_{0}}A(s)\ \mathrm{d}s,\qquad T_{0}=\beta^{-1/2}.
\end{equation}
$U_{1}(t)$ is generated by $\beta\overline{A}_{1}^{g}(t)$ with 
\begin{align}
 & \overline{A}_{1}^{(n)}\equiv\frac{1}{T_{0}}\int_{nT_{0}}^{(n+1)T_{0}}U_{0}^{-1}(t)\left[A(t)-\overline{A}_{0}^{g}(t)\right]U_{0}(t)\ \mathrm{d}t\label{eq:1.16}\\
 & U_{2}\equiv1+i\beta U_{1}^{-1}(t)\left(\int_{0}^{t}\left(A(t')-\overline{A}_{1}^{g}(t')\right)\ \mathrm{d}t\right)U_{1}(t)=O(\beta t);|t|\leq\frac{1}{\sqrt{\varepsilon}}.\label{eq:1.17}
\end{align}
Obtained by a model form transformation. 

Next to the leading order we get approximating $A_{NF}$ by 1 and
it implies that $\phi_{2}^{U}$ is a constant therefore using \eqref{eq:1.17}
and taking $U_{2}$ in the leading order implies 
\begin{equation}
\psi(t)=V(t)U_{0}(t)[U_{1}(t)\psi_{0}-U_{1}(t)\left(i\beta\int_{0}^{t}\left(A(t^{\prime})-A_{1}^{g}(t^{\prime})\right)dt^{\prime}\right)\psi(0)+O\left(\beta^{2}\right)]+O\left(\beta^{3/2}t\right)\label{eq:L2.25S}
\end{equation}
substituting $V\left(t\right)$, $U_{0}$ and $U_{1}$ implies 
\begin{equation}
\psi(t)=\mathcal{T}e^{-i\int_{0}^{t}H_{0}^{(g)}(t^{\prime})dt^{\prime}}\mathcal{T}e^{-i\beta\int_{0}^{t}\overline{A}_{0}^{g}(t^{\prime})dt'}\mathcal{T}e^{-i\beta\int_{0}^{t}\overline{A}_{1}^{g}(t^{\prime})dt^{\prime}}\left(U_{2}^{-1}\psi(0)\right)+O\left(\beta^{3/2}t\right)\label{eq:2.18aa}
\end{equation}

In order to simplify this, we use the following Lemma

\noindent \textbf{Lemma 2.2}

Let $A$ and $B$ be self adjoint. Then 
\begin{equation}
O\equiv\mathcal{T}e^{i\int_{0}^{t}A(s)ds}\mathcal{T}e^{i\int_{0}^{t}B(s)ds}=\mathcal{T}e^{+i\int_{0}^{t}C(t^{\prime})dt^{\prime}}\label{eq:L2.20}
\end{equation}
where 
\begin{equation}
C(t)=A(t)+\mathcal{T}e^{i\int_{0}^{t}A(s)ds}B(t)\mathcal{T}e^{-i\int_{0}^{t}A(s)ds}
\end{equation}

\noindent {Proof:} Differentiation of $O$ with respect to $t$
yields 
\begin{equation}
\frac{\partial}{\partial t}O=A(t)\mathcal{T}e^{i\int_{0}^{t}A(s)ds}\mathcal{T}e^{i\int_{0}^{t}B(s)ds}+\mathcal{T}e^{i\int_{0}^{t}A(s)ds}B(t)\mathcal{T}e^{i\int_{0}^{t}B(S)ds}=\left[A(t)+Te^{i\int_{0}^{t}A(s)ds}B(t)Te^{-i\int_{0}^{t}A(s)ds}\right]O\label{eq:2.23}
\end{equation}
identifying the term in the square brackets with $C(t)$ completes
the proof of the Lemma. \hfill{}\qedsymbol \\

\section{Landau Zener Majorana Example }

The Landau Zener problem is defined by the Schrödinger equation \cite{LZ1,LZ2,ASY,FMB,Vit1,Vit2,Wit,HJ,Ber1,Ber2}.
\begin{equation}
i\frac{\partial}{\partial t}\psi=H(\varepsilon t)\psi\label{eq:2.1-2}
\end{equation}
where $\psi$ is chosen to be a spinor and the Hamiltonian is 
\begin{equation}
H(\varepsilon t)=\begin{bmatrix}\varepsilon t & B\\
B & -\varepsilon t
\end{bmatrix}=\varepsilon t\sigma_{z}+B\sigma_{x}\label{eq:2.2-2}
\end{equation}
where $\sigma_{x},\sigma_{y}$ and $\sigma_{z}$ are the Pauli matrices.
This is a two state system (or its finite dimensional generalizations)
The adiabaticity parameter is $\epsilon$. Therefore in the present
work it will be assumed to be small. The time-averaged Hamiltonian
in the interval $nT_{0}\leq t\leq(n+1)T_{0}$ is given by 
\begin{equation}
\overline{H}_{n}=\begin{bmatrix}\varepsilon T_{0}\left(n+\frac{1}{2}\right) & B\\
B & -\varepsilon T_{0}\left(n+\frac{1}{2}\right)
\end{bmatrix}\label{eq:2.3-2}
\end{equation}
and in this interval $H_{0}^{(g)}$ of~\eqref{eq:1.9} takes this
value.

The propagator of~\eqref{eq:1.11b} is, for $nT_{0}<t<(n+1)T_{0}$,
given by 
\begin{equation}
V(t)=e^{-i\overline{H}_{N_{t}}\left(t\right)}f_{n-1}\cdots f_{1}f_{0}\label{eq:2.4}
\end{equation}
where 
\begin{equation}
f_{n}\equiv e^{-i\overline{H}_{n}T_{0}}.\label{eq:2.5}
\end{equation}

The effective Hamiltonian of Equation~\eqref{eq:1.13}, is 
\begin{equation}
A(t)=V(t)^{-1}\left[H(\varepsilon t)-H_{0}^{(g)}(t)\right]V(t).\label{eq:2.6}
\end{equation}

In the interval $N_{t}T_{0}<t<(N_{t}+1)T_{0}$ the term in the square
brackets is 
\begin{align}
H(\varepsilon t)-\overline{H}_{N_{t}} & =\begin{bmatrix}\varepsilon t & B\\
B & -\varepsilon t
\end{bmatrix}-\begin{pmatrix}\varepsilon T_{0}\left(N_{t}-\frac{1}{2}\right) & +B\\
B & -\varepsilon T_{0}\left(N_{t}-\frac{1}{2}\right)
\end{pmatrix}\\
 & =a_{N_{t}}(t)\begin{bmatrix}1 & 0\\
0 & -1
\end{bmatrix}=a_{N_{t}}(t)\sigma_{z}\label{eq:2.7}
\end{align}
where 
\begin{equation}
a_{N_{t}}(t)=t-\varepsilon T_{0}\left(N_{t}+\frac{1}{2}\right).\label{eq:2.8}
\end{equation}
since it must commute with $\sigma_{z}$ at time zero. Therefore
the equation for $\widetilde{\psi}(t)$ takes the form 
\begin{equation}
i\frac{\partial\tilde{\psi}}{\partial t}=\sqrt{\varepsilon}A(t)\tilde{\psi}=\sqrt{\varepsilon}a(t)\sigma_{z}\tilde{\psi}.\label{eq:7*}
\end{equation}
This equation can be easily solved, and one can see that it does not
produce any transition. Hence, the entire relevant dynamics is determined
by the piecewise constant (in time) averaged Hamiltonian.

We turn now to the calculation of $V(t)$ of (3.7) 

For this purpose we use the relation $T_{0}\varepsilon^{2}=1$ and
write 
\begin{align}
f_{n} & =e^{-i\left[\varepsilon T_{0}^{2}\left(n+\frac{1}{2}\right)\sigma_{z}+BT_{0}\sigma_{x}\right]}\ =e^{i\left[\sigma_{z}+\frac{\sigma_{x}}{\varepsilon T_{0}(n+\frac{1}{2})}\right]\left(n+\frac{1}{2}\right)}\label{eq:2.14}\\
\text{we denote }V(t=nT_{0}) & =V_{n}\label{eq:2.15}\\
V_{n+1} & =f_{n}V_{n}\label{eq:2.16}
\end{align}

A very useful formula is~\cite{LLQ} ,~p.~203 
\begin{equation}
f(a+\vec{b}\cdot\vec{\sigma})=\alpha+\vec{\beta}\cdot\vec{\sigma}\label{eq:2.17}
\end{equation}
where 
\begin{equation}
\alpha=\frac{1}{2}[f(a+b)+f(a-b)]
\end{equation}
\begin{equation}
\vec{\beta}=\frac{\vec{b}}{2b}[f(a+b)-f(a-b)].\label{eq:2.18}
\end{equation}
In our case $f\left(x\right)=e^{-ix}$and 
\begin{align}
a_{n}=0\quad{\rm {and}\ \vec{b_{n}}} & =\left(BT_{0},0,\varepsilon T_{0}^{2}\left(n+\frac{1}{2}\right)\right)\\
b_{n} & =T_{0}\sqrt{B^{2}+\varepsilon^{2}T_{0}^{2}\left(n+\frac{1}{2}\right)^{2}}\label{eq:3.13-shmuel}\\
\alpha_{n}=\cos b_{n}\quad\vec{\beta}_{n} & =-i\left(BT_{0},0,\varepsilon{T_{0}}\left(n+\frac{1}{2}\right)\right)\frac{\sin b_{n}}{b_{n}}\\
 & =\frac{-i\sin b_{n}}{T_{0}}\left(\frac{BT_{0}}{\sqrt{(B)^{2}+{\varepsilon^{2}T_{0}}^{2}\left(n+\frac{1}{2}\right)^{2}}},0,\frac{\varepsilon{T_{0}}^{2}\left(n+\frac{1}{2}\right)}{\sqrt{B^{2}+{\varepsilon^{2}T_{0}}^{2}\left(n+\frac{1}{2}\right)}}\right)\\
\\
\end{align}

\begin{align}
 & f_{n}=\cos b_{n}+\frac{-i\sin b_{n}}{b_{n}}\left(BT_{0}\cdot\sigma_{x}+\varepsilon{T_{0}}^{2}\left(n+\frac{1}{2}\right)\sigma_{z}\right)\label{eq:2.29}
\end{align}

Let us consider three different domains 
\begin{align}
\text{\ I }\quad T_{0} & \varepsilon n\gg B\label{eq:2.31}\\
\text{\ \ II }\quad T_{0} & \varepsilon n\sim B\label{eq:2.32}\\
\text{III }\quad T_{0} & \varepsilon n\ll B\label{eq:2.33}
\end{align}

First we obtain the expression for $f_{n,f_{n+1}},$where $f_{n}$is
given by \ref{eq:2.4} with 

We begin by recalling the definitions:

\begin{align}
 & \quad T_{0}\equiv\frac{1}{\sqrt{\varepsilon}}\\
b_{n} & =\epsilon T_{0}^{2}\left(n+\frac{1}{2}\right)\sqrt{1+\frac{B^{2}}{\epsilon^{2}T_{0}^{2}\left(n+\frac{1}{2}\right)^{2}}}
\end{align}
In domain $I$ 
\begin{align}
\frac{\varepsilon T_{0}\left(n+\frac{1}{2}\right)}{B}\gg1
\end{align}
Therefore 
\begin{equation}
b_{n}=\epsilon T_{0}^{2}+\frac{B^{2}}{2\epsilon\left(n+\frac{1}{2}\right)}+O\left(\frac{B^{2}}{\epsilon^{2}T_{0}^{2}n}\right)^{2}n
\end{equation}
For future use, we write $f_{n}$ in the following forms 
\begin{align}
f_{n} & \equiv e^{-i\left(\varepsilon{T_{0}}^{2}\left(n+\frac{1}{2}\right)\sigma_{z}+BT_{0}\sigma_{x}\right)}\\
 & =e^{-iT_{0}\left(\sqrt{\varepsilon}\left(n+\frac{1}{2}\right)\sigma_{z}+B\sigma_{x}\right)}\\
 & =e^{-iT_{0}\sqrt{\varepsilon}\left(n+\frac{1}{2}\right)\left[\sigma_{z}+\frac{B}{\sqrt{\varepsilon}\left(n+\frac{1}{2}\right)}\sigma_{x}\right]}\\
f_{n} & \equiv e^{-i\sqrt{\varepsilon}T_{0}\left(n+\frac{1}{2}\right)\left[\sigma_{z}+\varepsilon_{n}\sigma_{x}\right]}\\
 & \text{with }\epsilon_{n}=\frac{B}{\sqrt{\epsilon}\left(n+\frac{1}{2}\right)}\label{eq:10-shmuel}\\
\end{align}
$f_{n}$ of \ref{eq:2.29} can be written as 
\begin{equation}
f_{n}=\left[\cos b_{n}-\frac{i}{b_{n}}\sin b_{n}\varepsilon T_{0}^{2}\left(n+\frac{1}{2}\sigma_{z}\right)\right]\times\left[1-\left[\cos b_{n}-\frac{i}{b_{n}}\sin b_{n}T_{0}^{2}(n+\frac{1}{2}\right)\sigma_{z}\right]^{-1}\frac{i}{b_{n}}\sin b_{n}BT_{0}\sigma_{x}
\end{equation}
where $b_{n}$ is given by \eqref{eq:3.13-shmuel}. We define now
$b_{n}^{'}$ by 
\begin{equation}
e^{-ib'_{n}\sigma_{z}}\equiv\left[\cos b_{n}-\frac{i}{b_{n}}\sin b_{n}\varepsilon T_{0}^{2}\left(n+\frac{1}{2}\right)\sigma_{z}\right]
\end{equation}
then 
\begin{align}
f_{n} & =\overline{e}^{ib'_{n}\sigma_{z}}\left[I+e^{ib'_{n}\sigma_{z}}i\frac{BT_{0}}{b_{n}}\sin b_{n}\sigma_{x}\right]
\end{align}
where its main part is given by 
\[
F\left(b_{n}^{'}\right)=e^{-ib_{n}^{'}\sigma_{z}}
\]
 and $\epsilon_{n}$ is given by \ref{eq:10-shmuel}. note that $b'_{n}$
is defined through $f_{n}$ so that approximately for large $n$ $e^{-ib_{n}\sigma_{z}}$
or $b'_{n}\sim b_{n}$. The reason it can be done is that $b_{n}$
grows as $n$. We introduce the small correction. 
\begin{align}
e^{+ib'_{n}\sigma_{z}}\frac{BT_{0}}{b_{n}}\sin b_{n}\sigma_{x}\equiv\varepsilon_{n}^{X}.
\end{align}

We define $t_{n}=T_{0}\sqrt{\varepsilon}\left(n+\frac{1}{2}\right)$
and turn to the calculation of the product, using Lemma 2.2 (Eq \ref{L2.23},\ref{eq:L2.25S}
) one finds 

\begin{align}
f_{n}f_{n+1} & =e^{-it_{n}\left(\sigma_{z}+\varepsilon_{n}\sigma_{x}\right)}e^{-it_{n+1}\left(\sigma_{z}+\varepsilon_{n+1}\sigma_{x}\right)}\label{eq:fnfn1}\\
 & =e^{-it_{n}\left(\sigma_{z}+\varepsilon_{n}\sigma_{x}\right)}e^{-it_{n}\left(\frac{t_{n+1}}{t_{n}}\sigma_{z}+\frac{t_{n+1}}{t_{n}}\varepsilon_{n+1}\sigma_{x}\right)}\\
 & =\exp i\int_{0}^{t_{n}}C_{n}(s)ds
\end{align}
where 
\begin{equation}
C_{n}(s)=\sigma_{z}+\varepsilon_{n}\sigma_{x}+e^{-is\left(\sigma_{z}+\varepsilon_{n}\sigma_{x}\right)}\left[\frac{t_{n+1}}{t_{n}}\sigma_{z}+\frac{t_{n+1}}{t_{n}}\varepsilon_{n+1}\sigma_{x}\right]e^{+is\left(\sigma_{z}+\varepsilon_{n}\sigma_{x}\right)}
\end{equation}
Using 3.20 - 3.23 one finds
\begin{align}
e^{-is\left(\sigma_{z}+\varepsilon_{n}\sigma_{x}\right)} & =\frac{1}{2}e^{-is\sqrt{1+{\varepsilon_{n}}^{2}}}+\frac{1}{2}e^{is\sqrt{1+{\varepsilon_{n}}^{2}}}\\
 & -\left(i\sin s\sqrt{1+{\varepsilon_{n}}^{2}}\right)\frac{1}{\sqrt{1+{\varepsilon_{n}}^{2}}}\left(\sigma_{z}+\varepsilon_{n}\sigma_{x}\right)\\
 & =\cos s\sqrt{1+{\varepsilon_{n}}^{2}}-\frac{i}{\sqrt{1+{\varepsilon_{n}}^{2}}}\sin\sqrt{1+{\varepsilon_{n}}^{2}}s\left(\sigma_{z}+\varepsilon_{n}\sigma_{x}\right)
\end{align}
We turn to define $\varepsilon'_{n}$ and $\tilde{\varepsilon}_{n}$
\begin{align}
e^{-is\sqrt{1+{\varepsilon{}_{n}}^{2}}\sigma_{z}} & \left(I+e^{is\sqrt{1+{\varepsilon_{n}^{3}}^{2}}\sigma_{z}}\frac{-i}{\sqrt{1+{\varepsilon_{n}}^{2}}}\sin\sqrt{1+{\varepsilon_{n}}^{2}}\varepsilon_{n}\sigma_{x}\right)\label{eq:*12}\\
 & \equiv e^{-is\sqrt{1+{\varepsilon'_{n}}^{2}}\sigma_{z}}\left(I+\tilde{\varepsilon}_{n}(s)\sigma_{x}\right)\\
\nonumber 
\end{align}
and it satisfies 
\[
\epsilon_{n}^{'}\approx\epsilon_{n}
\]
We now turn to calculate the leading term with the help of \eqref{eq:*12}
\begin{align}
e^{-is\left(\sigma_{z}+\varepsilon_{n}\sigma_{x}\right)}\frac{t_{n+1}}{t_{n}}\sigma_{z}e^{is\left(\sigma_{z}+\varepsilon_{n}\sigma_{x}\right)}\\
=e^{-is\sqrt{1+{\varepsilon'_{n}}^{2}}\sigma_{z}}\left(I+\tilde{\varepsilon}_{n}\sigma_{x}\right)\left(\frac{t_{n+1}}{t_{n}}\right)\sigma_{z}\left(I-\tilde{\varepsilon}_{n}\sigma_{x}\right)e^{is\sqrt{1+{\varepsilon_{n}}^{2}}\sigma_{z}}.
\end{align}
Now \eqref{eq:fnfn1} takes takes form
\begin{equation}
f_{n}f_{n+1}=\exp\left[-i\int_{0}^{t_{n}}C_{n}(s)ds\right]
\end{equation}
with

\begin{align}
\\
C_{n}(s) & =\sigma_{z}+\varepsilon_{n}\sigma_{x}+\frac{t_{n+1}}{t_{n}}\sigma_{z}+\frac{t_{n+1}}{t_{n}}F\left(s\sqrt{1+{\varepsilon'_{n}}^{2}}\right)\varepsilon_{n+1}\sigma_{x}F^{\ast}\left(s\sqrt{1+{\varepsilon'_{n}}^{2}}\right)+\\
 & \quad F\left(s\sqrt{1+{\varepsilon'_{n}}^{2}}\right)\tilde{\varepsilon}_{n}(s)\frac{t_{n+1}}{t_{n}}\left(-2i\sigma_{y}\right)F^{\ast}\left(s\sqrt{1+{\varepsilon'_{n}}^{2}}\right)+O\left({\tilde{\varepsilon}_{n}}^{2}(s)\right)\sigma_{z}\\
 & =\left[1+\frac{t_{n+1}}{t_{n}}+O\left(\tilde{{\varepsilon}_{n}}^{2}\right)\right]\sigma_{z}+\varepsilon_{n}\sigma_{x}+iF\left(s\sqrt{1+{\varepsilon'_{n}}^{2}}\right)\tilde{\varepsilon}_{n}\sigma_{y}F^{\ast}\left(s\sqrt{1+{\varepsilon'_{n}}^{2}}\right)\left(\frac{-t_{n+1}}{t_{n}}\right)\\
 & \quad+F\left(s\sqrt{1+{\varepsilon'_{n}}^{2}}\right)\tilde{\varepsilon}_{n+1}\sigma_{x}\left(\frac{+t_{n+1}}{t_{n}}\right)F^{\ast}\left(s\sqrt{1+{\varepsilon'_{n}}^{2}}\right)
\end{align}
so, 
\begin{align}
f_{n}f_{n+1} & =\mathcal{T}\text{exp}\left[~^{(-i)\left[t_{n}+t_{n+1}+O_{n}(\varepsilon)\right]\sigma_{z}+i\frac{t_{n}\varepsilon_{n}}{t_{n}+t_{n+1}}\sigma_{x}\left(t_{n}+t_{n+1}\right)}\right.\\
 & \cdot(t_{n}+t_{n+1}+O(\varepsilon))\frac{t_{n+1}}{t_{n}\left(t_{n}+t_{n+1}+O_{n}\right)}\\
 & \left.\quad\int_{0}^{t_{n}}F\varepsilon_{n+1}\sigma_{x}F^{\ast}ds+2i\frac{t_{n+1}}{t_{n}}\int_{0}^{t_{n}}F\left(s\sqrt{1+{\varepsilon'_{n}}^{2}}\right)\sigma_{y}F^{*}\left(s\sqrt{1+{\varepsilon'_{n}}^{2}}\right)ds\right]\\
 & =\mathcal{T}\exp(-i\left\{ \left[t_{n}+t_{n+1}+O(\varepsilon_{n})\right]\left[\sigma_{z}+\frac{\varepsilon_{n}t_{n}}{{O}(\varepsilon_{n})+t_{n}+t_{n+1}}\sigma_{x}+\frac{t_{n+1}}{{O}(\varepsilon_{n})+t_{n}+t_{n+1}}\frac{1}{t_{n}}\int_{0}^{t_{n}}G(s)\cdots\right]\right\} )
\end{align}
where 
\begin{align}
G(s)=\cdots=\frac{1}{t_{n}}\int_{0}^{t_{n}}2iF\varepsilon_{n}(s)\sigma_{y}F^{*}ds+\frac{1}{t_{n}}\int_{0}^{t_{n}}F\varepsilon_{n+1}\sigma_{x}F^{\ast}ds.
\end{align}
The expansion has two parts, the integral one which is time dependent
part denoted by $G\left(s\right)$, the rest. Therefore the following
Lemme is useful 

\emph{Lemma}3.1. \smallskip{}

Let $U(t)=\mathcal{T}e^{-iAt-i\int_{0}^{t}B(s)ds}$ with $B(t)=U(t)BU^{\ast}(t)$.
Then $\frac{dU}{dt}=(-iA-i\beta(t))F$.

If $U=e^{-iAt}W(t)$ then 
\begin{align}
\dot{W} & =-ie^{iAt}B(t)e^{-iAt}W
\end{align}

In our case 
\begin{align}
A & =\sigma_{z}\left(1+\frac{t_{n+1}}{t_{n}}\right)a\left(O\left(\epsilon_{n}\right)+O\left(\tilde{\epsilon_{n}}\right)\right)\ ,\ B(t)=C_{n}F\left(t\sqrt{1+{\varepsilon'_{n}}^{2}}\right)\tilde{\varepsilon}_{n}(s)\sigma_{y}F^{\ast}\\
C_{n} & =\frac{-t_{n+1}}{t_{n}}\cdot\frac{1}{t_{n}}\\
 & \Rightarrow\dot{W}\approx iC_{n}\tilde{\varepsilon}_{n}(t)\sigma_{y}W.
\end{align}

Where we used that $e^{-At}F\approx1$, and it is understood that
$t=t_{n}$ is taken. Integrating, we get 
\begin{align}
W(t_{n}) & =e^{-iC_{n}\sigma_{y}\varepsilon_{n}\int_{0}^{t_{n}}\frac{1}{\sqrt{1+{\varepsilon_{n}}^{2}}}\sin\sqrt{1+{\varepsilon_{n}}^{2}}sds}\\
 & =e^{-i\frac{t_{n+1}}{t_{n}}\sigma_{y}\left(\cos\sqrt{1+{\varepsilon_{n}}^{2}}t_{n}-1\right)\frac{\varepsilon_{n}}{1+{\varepsilon_{n}}^{2}}\frac{1}{t_{n}}}.
\end{align}

The integral is bounded therefore its contribution will turn out to
be small. We arrived at 
\begin{align*}
f_{n}f_{n+1} & \sim\exp i\int_{0}^{t_{n}+t_{n+1}+O_{n}(\varepsilon)}\tilde{D}_{n}ds
\end{align*}
with 
\begin{align}
\tilde{D}_{n}\sim\sigma_{z}+\left(\frac{t_{n}}{t_{n}+t_{n+1}+O(\varepsilon_{n})}\varepsilon_{n}\sigma_{x}+\frac{t_{n+1}}{t_{n}+t_{n+1}+O_{n}(\varepsilon)}\varepsilon_{n+1}\sigma_{x}\right).
\end{align}
What we did actually is extended the domain of integration from $t_{n}$
to $t_{n}+t_{n+1}$ prefactor of the leading term was set to remain
the same while the corrections are different. Hence to the leading
order
\begin{equation}
f_{n}f_{n+1}\approx e^{-i\left(t_{n}+t_{n+1}\right)\sigma_{z}}
\end{equation}
 Iterating the process $\tilde{D_{n}}$ is replaced by $\tilde{D_{n+1}}$
so that The correction term decreases compared to the $\sigma_{z}$
term. Continuing the precess results in 
\begin{align}
f_{n}f_{n+1}f_{n+2} & =\exp i\int_{0}^{t_{n}+t_{n+1}+t_{n+2}+O_{n+1}(\varepsilon)}\tilde{D}_{n+1}(s)ds\quad\tilde{D}_{n+1}(s)\sim\\
 & \sigma_{z}+\sigma_{x}\left(\frac{t_{n}}{t_{n}+t_{n+1}+t_{n+2}+O_{n+1}(\varepsilon)}\varepsilon_{n}+\frac{t_{n+1}}{(\cdots)}\varepsilon_{n+1}+\frac{t_{n+2}}{(\cdots)}\varepsilon_{n+2}\right).\label{eq:3.39}
\end{align}

If we iterate the process starting from \ref{eq:fnfn1}, to leading
order one finds, 
\begin{equation}
f_{N_{0}}f_{2}f_{3}...f_{N}\approx e^{-i\left(t_{N_{0}}....t_{N}\right)\sigma_{z}}+C.C
\end{equation}
where $N_{0}$ is the minimal $n$ in Domain I.

The first correction in \ref{eq:3.39} is of order $\frac{1}{N^{2}}$
since $t_{n}$ is of order $N_{0}$ while $\epsilon_{n}$ is of order
$\frac{1}{N_{0}}$. Therefore in the limit $N\rightarrow\infty$ it
vanishes

The $\varepsilon_{n}$ contributions decay with iterations, since
each one is proportional to one time $t_{n}$, out of overall time
$t_{1}+t_{2}+\cdots+t_{N}$.

Therefore $V_{n}\approx e^{-iN\sigma_{z}}$ commuting with $\sigma_{z}$,
based on \eqref{eq:1.8},\eqref{eq:1.13} we find 
\begin{equation}
i\frac{\partial\tilde{\psi}}{\partial t}=\sqrt{\varepsilon}A(t)\tilde{\psi}=\sqrt{\varepsilon}a(t)\sigma_{z}\tilde{\psi}.\label{eq:7*-1}
\end{equation}

\section*{Domain III}

In this Domain $f_{n}$ takes the form 
\begin{equation}
f_{n}=e^{-iBT_{0}\sigma_{x}}\left(I+\text{corrections}\right)
\end{equation}
and this form holds for $n+\frac{1}{2}\ll BT_{0}$ where $T_{0}^{2}\epsilon=1$
was used. 

A procedure similar to the one that was used in domain I leads to
\begin{equation}
V_{N}\approx e^{-iNBT_{0}^{2}\sigma_{x}}
\end{equation}
In this case A of \eqref{eq:1.13} is 
\begin{equation}
A=\sqrt{\epsilon}a_{n}\left(t\right)e^{iNBT_{0}\sigma_{x}}\sigma_{z}e^{-iNBT_{0}\sigma_{x}}\label{eq:ham **}
\end{equation}
 where $a_{n}$ is given by \eqref{eq:2.8}.

In this Domain therefore, to the leading order $\psi$ develops with
a Hamiltonian proportional to $\sigma_{x}$ while $\tilde{\psi}$
evolves by the Hamiltonian \eqref{eq:ham **}

\section*{Domain II}

In this domain there is no simple expression for $V$ of \ref{eq:1.11bb},\ref{eq:1.11b}
and it reduces to products of matrices

\section{Adiabatic Theorem}

In this section, we give a new proof of the adiabatic theorem in the
gapless case. Furthermore, we prove an ergodic theorem that holds
uniformly for the dynamics generated by adiabatic Hamiltonians. In
fact, the ergodic theorem implies the adiabatic theorem. We start
with the formulation of the theorem. 

\noindent \textbf{\begin{theorem}\end{theorem} The Adiabatic Theorem}

Let $\left\{ H(s)\left\{ |s\in[0,1]\right\} \right\} $ be a family
of bounded self-adjoint operators on $\left\{ \mathcal{H}\right\} $
- a separable Hilbert space.

Assume that 
\begin{enumerate}
\item i) $H(s),\ \dot{H}(s),\ddot{H}(s)$ are bounded and continuous uniformly
in $s\in[0,1]$ (where $\dot{f}\equiv\frac{df}{dt}$). 
\item ii) $H(s)$ has an eigenvector $\psi_{0}(s)$, with projection $P_{0}(s)$
on the bound state we focus on. It will be denoted $\psi_{0}$ here
and in in what follows.

\begin{enumerate}
\item $P(s),\dot{P}(s),\ \ddot{P}(s)$ are bounded and continuous, uniformly
in $s\in[0,1]$. 
\item There is a minimal distance between the eigenvalue $\lambda_{0}(s)$
of $\psi_{0}(s)$ and any other eigenvalue of $H(s)$. $\lambda_{0}(s)$
may be embedded in the continuous spectrum of $H(s)$. 
\end{enumerate}
\end{enumerate}
Then, if the initial condition of the Schrödinger equation $(\varepsilon>0$,
small) 
\begin{eqnarray}
i\frac{\partial\psi_{\varepsilon}}{\partial t}=H(\varepsilon t)\psi_{\varepsilon}
\end{eqnarray}
is $\psi(0)=\psi_{0}(t)$ at $t=0$, we have that for all $0<t\leq\frac{1}{\epsilon}.$
\begin{eqnarray}
\|\psi_{\varepsilon}(t)-e^{i\theta(t)}\psi_{0}(\varepsilon t)\|=o_{\varepsilon}(1)\ .
\end{eqnarray}

Here, ${o}_{\varepsilon}(1)$ stands for a function that vanishes
as $\varepsilon\rightarrow0$ and $\theta\left(t\right)$ is a real
valued function of $t$

\noindent \textbf{Proof}

The proof follows the strategy of Kato\cite{key-1}, which reduces
the problem to estimating the size of an appropriate wave operator.
Let $U_{K}(t)$ stand for the associated Kato dynamics
\begin{equation}
i\frac{dU_{k}\left(t\right)}{dt}=K\left(t\right)U_{k}\left(t\right)
\end{equation}
\begin{equation}
K\left(t\right)\equiv H\left(\epsilon t\right)+i\epsilon\left[\dot{P}_{0}\left(\epsilon t\right),P_{0}\left(\epsilon t\right)\right]
\end{equation}
 where $\dot{P}_{0}(\varepsilon t)$ stands for $\frac{\partial P_{0}(\mu)}{\partial\mu}|_{\mu=\epsilon t}$
and $K\left(t\right)$ is the Hamiltonian of the Kato dynamics.

The main property of the Kato dynamics is that it evolves $P_{0}(0)$
to $P_{0}(\varepsilon t)$, unitarily: 
\begin{equation}
P_{0}\left(\epsilon t\right)U_{k}\left(t\right)=U_{k}P_{0}\left(0\right)\label{eq:blaaa}
\end{equation}
 Then, the proof of the theorem follows from showing that 
\begin{equation}
\|U^{*}(t)U_{K}(t)P_{0}(0)-P_{0}(0)\|\leq o_{\epsilon}(1)\ \ {\rm for}\ 0\leq t\leq1/\varepsilon.\label{eq:4.5}
\end{equation}
Writing the above as the integral of the derivative (Cook's method),
the problem is reduced to proving that (see Eq. \ref{eq:4.15}) 
\begin{eqnarray}
\|\varepsilon\int_{0}^{1/\varepsilon}U^{*}(t)\left[\dot{P}_{0}(\varepsilon t),\ P_{0}(\varepsilon t)\right]U_{K}(t)P_{0}(0)dt\|=o_{\epsilon}(1).\label{eq:4.6}
\end{eqnarray}

A key observation for the proof is the following proposition, a consequence
of the formula for integration by parts:\begin{prp}\label{4.1}\end{prp}

Let $A(t)$, $B(t)$ be families of bounded smooth operators and $0<\epsilon\ll1$;
then, 
\begin{eqnarray}
\epsilon\int_{0}^{1/\varepsilon}A(t)B(t)dt=\epsilon\left(\int_{0}^{s}A(s')ds'\right)B(s)|_{s=1/\varepsilon}\label{eq:4-7}\\
-\epsilon\int_{0}^{1/\varepsilon}\left(\int_{0}^{s}A(s')ds'\right)\frac{dB(s)}{ds}ds\nonumber 
\end{eqnarray}
Furthermore, assume that 
\begin{equation}
\|\epsilon\int_{0}^{1/\varepsilon}A(s)ds\|={o}_{\varepsilon}(1)
\end{equation}
and 
\begin{equation}
\|\frac{dB(s)}{ds}\|={O}(\varepsilon),
\end{equation}
then 
\begin{eqnarray}
\|\epsilon\int_{0}^{1/\varepsilon}A(t)B(t)dt\|={o}_{\varepsilon}(1).\label{eq:4.7 fish}
\end{eqnarray}
\begin{proof}

The proof starts with integration by parts.

To prove \ref{eq:4.7 fish}, we write the second term on the RHS of
Eq. \eqref{eq:4-7} as 
\begin{eqnarray}
\epsilon\int_{0}^{1/\epsilon}\left(\int_{0}^{1/\sqrt{\varepsilon}}A(s')ds'\right)\frac{dB(s)}{ds}ds+\epsilon\int_{0}^{1/\varepsilon}\left(\int_{1/\sqrt{\varepsilon}}^{s}A(s')ds'\right)\frac{dB(s)}{ds}ds.\label{eq:4.12}
\end{eqnarray}
The norm of the first term in Eq \eqref{eq:4.12} is bounded by 
\begin{equation}
\sqrt{\varepsilon}\int_{0}^{1/\varepsilon}\|\sqrt{\varepsilon}\int_{0}^{1/\sqrt{\varepsilon}}A(s')ds'\|\|\frac{dB(s)}{ds}\|ds\leq\sqrt{\varepsilon}\int_{0}^{1/\varepsilon}O(\varepsilon){o}_{\sqrt{\varepsilon}}(1)ds\leq\sqrt{\varepsilon}{o}_{\sqrt{\varepsilon}}(1).\label{eq:412fish}
\end{equation}

The second term in \eqref{eq:412fish} is bounded by 
\begin{equation}
\int_{0}^{1/\varepsilon}O(\varepsilon)ds\sup_{\frac{1}{\varepsilon}\geq s\geq1/\sqrt{\varepsilon}}\|\epsilon\int_{1/\sqrt{\varepsilon}}^{s}A(s')ds'\|\leq
\end{equation}
\begin{align*}
\sup_{\frac{1}{\sqrt{\epsilon}}\leq s\leq\frac{1}{\epsilon}}O\left(\epsilon\right)\left[\left\Vert \epsilon\int_{0}^{1/\sqrt{\epsilon}}A\left(s'\right)ds'\right\Vert +\left\Vert \epsilon\int_{0}^{s}A\left(s'\right)ds'\right\Vert \right]\\
\leq & \epsilon\frac{1}{\sqrt{\epsilon}}+\sup_{\frac{1}{\sqrt{\epsilon}}<s<\frac{1}{\epsilon}}\frac{O\left(\epsilon\right)}{s}\left\Vert \int_{0}^{s}A\left(s'\right)ds'\right\Vert \leq O\left(\epsilon\right)\sigma_{\epsilon}\left(1\right)
\end{align*}
 This completes proof of proposition 1

\end{proof}

If we now write the integrand of Eq. \eqref{eq:4.6} as (using the
unitarity of $U$) 
\begin{equation}
U^{*}(t)\left[\dot{P}_{0}(\varepsilon t),\ P_{0}(\varepsilon t)\right]P_{0}(\varepsilon t)U(t)U^{*}(t)U_{K}(t)P_{0}
\end{equation}
and notice that 
\begin{eqnarray}
i\frac{d}{dt}U^{*}(t)U_{K}(t)P_{0}\left(0\right)=U^{*}(t)\left[H(\varepsilon t)-H(\varepsilon t)+i\epsilon\left[\dot{P}_{0},P_{0}\right]\right]U_{K}(t)P_{0}\left(0\right)\\
=\epsilon U^{*}(t)\left[\dot{P}_{0},P_{0}\right]U_{K}(t)P_{0}\left(0\right)=O(\varepsilon),\label{eq:4.15}
\end{eqnarray}
then (\ref{eq:4.5}) follows from the above proposition if we prove
that 
\begin{eqnarray}
\|\epsilon\int_{0}^{1/\varepsilon}U^{*}(t)\left[\dot{P}_{0}(\varepsilon t),{P}_{0}(\varepsilon t)\right]P_{0}(\varepsilon t)U(t)dt\|={o}_{\varepsilon}(1).\label{eq:4.8 fish}
\end{eqnarray}
Now, using that $P_{0}(s)^{2}=P_{0}(s)$ it follows that 
\begin{eqnarray}
 & \dot{P}_{0}P_{0}+P_{0}\dot{P}_{0}=\dot{P}_{0},\ \ {\rm {so}}\\
 & P_{0}\dot{P}_{0}P_{0}+P_{0}\dot{P}_{0}=P_{0}\dot{P}_{0}\\
{\rm or}\  & P_{0}\dot{P}_{0}P_{0}=0,\ \ (\dot{P}_{0}P_{0}-P_{0}\dot{P}_{0})P_{0}=\dot{P}_{0}P_{0}\label{eq:4.20}
\end{eqnarray}
using the fact that 
\[
\sum_{j}P_{j}+P_{c}=I
\]
where $\left\{ P_{j}\right\} $ are the bound states and $P_{c}$
is the projection on the continuum combined with (\ref{eq:4.20})
we find that the integrand of (\ref{eq:4.8 fish}) is 
\begin{eqnarray}
U^{*}(t)\left\{ \sum_{j\neq0}P_{j}(\varepsilon t)\dot{P}_{0}(\varepsilon t)P_{0}(\varepsilon t)+P_{c}(H(\varepsilon t))\dot{P}_{0}(\varepsilon t)P_{0}(\varepsilon t)\right\} U(t).\label{eq:4.9 fish}
\end{eqnarray}
:\begin{lm}\end{lm}

Let $H_{1},$ $H_{2}$ bounded self adjoint . operators, and such
that $\|H_{1}-H_{2}\|\leq\delta$. Suppose 
\begin{eqnarray}
\|\frac{1}{T}\int_{0}^{T}A_{1}(s)ds\|\equiv\|\frac{1}{T}\int_{0}^{T}e^{iH_{1}t}Ae^{-itH_{1}}dt\|={o}_{T}(1).\label{eq:4.21}
\end{eqnarray}
(that is, ${o}_{T}(1)\rightarrow0$ as $T\rightarrow\infty)$, where
$A_{j}=e^{iH_{j}t}Ae^{-iH_{j}t}$ with $j=1,2$ . Then 
\begin{eqnarray}
\|\frac{1}{T}\int_{0}^{T}A_{1}(s)ds-\frac{1}{T}\int_{0}^{T}A_{2}(s)ds\|\leq\delta{o}_{T}(1)\label{4.25}
\end{eqnarray}
where ${o}_{T}(1)$ is as above.

\noindent \textbf{Remark}

The Lemma does not require that the ergodic bound ( Eq. \eqref{eq:4.21})
holds for $H_{2}$. It will follow if $\delta\lesssim1$, and with
the {\em same} rate function ${o}_{T}(1)$.

\noindent \textbf{Proof}

The difference in (\ref{4.25}) can be written as 
\begin{eqnarray}
 & \|\frac{1}{T}\int_{0}^{T}\left[e^{itH_{2}}A\left(e^{-itH_{1}}-e^{-itH_{2}}\right)+\left(e^{itH_{1}}-e^{itH_{2}}\right)Ae^{-itH_{1}}\right]dt\|\\
 & =\|\frac{1}{T}\int_{0}^{T}\left[A_{1}(t)\left(I-\Omega_{12}(t)\right)+\left(I-\Omega_{21}(t)\right)A_{2}(t)\right]dt\|\nonumber 
\end{eqnarray}
Where $\Omega_{ij}=e^{iH_{j}t}e^{-iH_{i}t}$. The result now follows
by applying the proposition, using our assumption that

\[
\|\frac{1}{T}\int_{0}^{T}A_{1}(t)dt\|=o_{T}(1),
\]
with similar relations for $A_{2}$ and 
\begin{eqnarray}
\|\frac{d}{dt}\Omega_{ij}(t)\|=\|\frac{d}{dt}e^{iH_{i}t}e^{-iH_{j}t}\|=\|e^{iH_{i}t}i(H_{i}-H_{j})e^{-iH_{j}t}\|\nonumber \\
~~~~~~~~~~~\leq\|H_{i}-H_{j}\|\leq\delta,\\
\nonumber 
\end{eqnarray}
for $i,j=1,2;\ \ i\neq j.$

\hfill{}\qedsymbol

We now turn back to slowly changing Hamiltonians. We saw that in this
case the solution $\psi(t)$ is given by ( see Sec.\ref{sec:Slowly-changing-potentials},
\ref{eq:L2.25S} and \ref{eq:2.18aa} ) 
\[
\psi(t)=U(t)\psi_{0}=V(t)\tilde{U}(t)\psi_{0}
\]
with
\[
U\left(t\right)=V\left(t\right)\cdot\tilde{U}\left(t\right)
\]

The evolution of the averaged Hamiltonian is given by $V\left(t\right)$
while the correction is given by $\tilde{U}\left(t\right)$. The adiabaticity
is reflected by mostly constant $\tilde{U}$ and more precisely
\[
\left\Vert \frac{d\tilde{U}\left(t\right)}{dt}\right\Vert \leq\text{const}\beta=\text{const}\sqrt{\epsilon}\ ,\ 0\leq t\leq\frac{1}{\epsilon}
\]
 and so, now let us define the averaged Hamiltonian 
\begin{eqnarray}
\bar{A}(T)\equiv\frac{1}{T}\int_{0}^{T}U^{*}(t)AU(t)dt\\
=\frac{1}{T}\int_{0}^{T}\tilde{U}^{*}(t)V^{*}(t)AV(t)\tilde{U}(t)dt\nonumber \\
\equiv\frac{1}{T}\int_{0}^{T}\tilde{U}^{*}(t)A_{V}(t)\tilde{U}(t)dt.\nonumber 
\end{eqnarray}
where
\begin{equation}
A_{V}(t)\equiv V^{*}(t)AV(t).
\end{equation}
To estimate $\bar{A}(T)$, we break the sum into $N=\left\{ \frac{1}{\sqrt{\varepsilon}}\right\} $
intervals of size $\frac{1}{\sqrt{\varepsilon}}$ where $\left\{ x\right\} $
is the closest integer to $x$.

We write:
\begin{eqnarray}
\epsilon\int_{0}^{\frac{1}{\varepsilon}}\tilde{U}^{*}(t)V^{*}(t)AV(t)\tilde{U}(t)dt\\
=\sqrt{\varepsilon}\sum_{j=0}^{N-1}\sqrt{\varepsilon}\int_{j\frac{1}{\sqrt{\varepsilon}}}^{(j+1)\frac{1}{\sqrt{\varepsilon}}}\tilde{U}^{*}(t)A_{V}(t)\tilde{U}(t)dt.\label{eq:4.29}
\end{eqnarray}
Therefore, if we can prove that 
\begin{eqnarray}
\sqrt{\varepsilon}\int_{j\varepsilon^{-\frac{1}{2}}}^{(j+1)\varepsilon^{-\frac{1}{2}}}\tilde{U}^{*}(t)A_{V}(t)\tilde{U}(t)dt\leq{o}_{\sqrt{\varepsilon}}(1),\label{eq:430}
\end{eqnarray}
\textbf{uniformly} in $j$, the ergodic theorem will follow for the
pair $(U(t),A)$.

By the proposition \ref{4.1} provided 
\begin{eqnarray}
\|\frac{d\tilde{U}(s)}{ds}\|\leq\sqrt{\varepsilon}=\beta\label{eq:4.15fish}
\end{eqnarray}
and using the proposition \ref{4.1} for $\tilde{U}$ and $\tilde{U^{*}}\ $the
result \eqref{eq:430} follows from 
\begin{eqnarray}
\|\sqrt{\varepsilon}\int_{j\frac{1}{\sqrt{\varepsilon}}}^{(j+1)\frac{1}{\sqrt{\varepsilon}}}A_{V}(s)ds\|={o}_{\sqrt{\varepsilon}}(1)\ {\rm (uniformly\ in}\ j)\label{eq:4.14 fish}
\end{eqnarray}
Since condition \eqref{eq:4.15fish} is satisfied by construction
of the averaging, it remains to verify the estimate \ref{eq:4.14 fish}
on the piecewise constant dynamics $V(t)$.

Next we prove: \begin{theorem} Uniform Ergodic Theorem \end{theorem}

Under our previous assumptions on the family of Hamiltonians $\{H(s)|0\leq s\leq1\}$,
suppose $A$ is a compact operator, and denote by $P_{c}(t)$ the
projection on the Hilbert subspace of continuos spectrum of $\mathcal{H}_{c}(H(\varepsilon t)).$
Then, 
\begin{eqnarray}
\|\frac{1}{T}\int_{0}^{T}U^{*}(t)AP_{c}(\epsilon t)U(t)dt\|={o}_{T}(1)
\end{eqnarray}
where $T=\frac{1}{\varepsilon}$, $U(t)$ is generated by the family
of Hamiltonians $H(\varepsilon t)$, for all $\varepsilon$ sufficiently
small.

\noindent \textbf{Proof}

By the discussion above Eqs (\ref{eq:4.29} - \ref{eq:4.15fish}),
we need only to show that 
\begin{eqnarray}
\|\sqrt{\varepsilon}\int_{j\varepsilon^{-\frac{1}{2}}}^{(j+1)\varepsilon^{-\frac{1}{2}}}V^{*}(t)AP_{c}(t)V(t)dt\|={o}_{\varepsilon}(1)\ {\rm {uniformly\ in}\ 0\leq j\leq\varepsilon^{-\frac{1}{2}}}
\end{eqnarray}
Now, recall that (Sec 2 and \cite{FS}) 
\begin{equation}
V(s)=V_{n}(s-nT_{0})V_{n-1}(T_{0})\cdots V_{1}(T_{0})
\end{equation}
where 
\begin{eqnarray}
V_{i}(s)=e^{-i\bar{H}_{i}\left(s\right)}\\
T_{0}=\frac{1}{\sqrt{\beta}}=\varepsilon^{-\frac{1}{4}}\ \ (\beta\equiv\sqrt{\epsilon}).
\end{eqnarray}
$\bar{H}_{i}$ is the averaged Hamiltonian in the time interval $i$.
So,we have to evaluate 
\begin{eqnarray}
 &  & S_{j}=\sqrt{\varepsilon}\int_{jT_{0}^{2}}^{(j+1)T_{0}^{2}}V^{*}\left(s\right)AP_{c}\left(s\right)V\left(s\right)ds\\
 &  & =\sqrt{\beta}\sum_{k=0}^{N-1}\sqrt{\beta}\int_{jT_{0}^{2}+kT_{0}}^{jT_{0}^{2}+(k+1)T_{0}}V_{1}^{*}(T_{0})\cdots V_{n_{jk}}^{*}(s-n_{jk}T_{0})AP_{c}(s)\cdot\\
 &  & V_{n_{jk}}\cdots V_{1}(T_{0})ds\nonumber 
\end{eqnarray}
We divided the j-th interval into $N=\left\{ T_{0}\right\} $ subintervals.
In each of these the averaged Hamiltonian is constant.

Taking the norm of the above, we get the bound on each term in the
sum over $j$ is : 
\begin{eqnarray}
\left|S_{j}\right| & \leq & \frac{1}{T_{0}^{2}}T_{0}\sup_{j}\|\int_{T_{jk}}^{T_{jk}+T_{0}}V_{n_{jk}}^{*}\left(s-n_{jk}T_{0}\right)AP_{c}(s)V_{n_{jk}}\left(s-n_{jk}T_{0}\right)ds\|\leq\label{eq:440}\\
 & \leq & \sup_{\begin{array}{c}
j,k\\
k\geq1
\end{array}}{o}_{T_{jk}}(1),
\end{eqnarray}
 $n_{j}$ is the index of the j,k interval starting at time $T_{jk}=jT_{0}^{2}+kT_{0}$.
The last inequality results from the ergodic theorem for systems with
time independent Hamiltonians \cite{key-1-avy_ref}

It is left to show uniformity of the estimate in $j$.
\[
\]

\noindent \textbf{Lemma $\bar{H}$}

\begin{eqnarray}
\bar{H}_{j}=\frac{1}{T_{0}^{2}}\int_{jT_{0}^{2}}^{(j+1)T_{0}^{2}}H(\varepsilon s)ds=H(y_{j})+\frac{\sqrt{\varepsilon}}{2}H'(y_{j})+\sup_{y}\|\ddot{H}(y)\|O_{H}(\varepsilon)
\end{eqnarray}
Where $O_{H}\left(\epsilon\right)$ is a bound operator with norm
$O\left(\epsilon\right)$

\noindent \textbf{Proof} 
\begin{eqnarray}
\bar{H}_{j} & = & \frac{1}{T_{0}^{2}}\frac{1}{\varepsilon}\int_{y_{j}}^{y_{j}+\sqrt{\varepsilon}}H(y)dy=\frac{1}{\sqrt{\varepsilon}}\int_{y_{j}}^{y_{i}+\sqrt{\varepsilon}}H(y)dy\\
 & = & \frac{1}{\sqrt{\varepsilon}}\left[\int_{y_{j}}^{y_{j}+\sqrt{\varepsilon}}H(y_{j})dy+\int_{y_{j}}^{y_{j}+\sqrt{\varepsilon}}\dot{H}(y_{j})(y-y_{j})dy\right.\\
 & + & \left.\int_{y_{j}}^{y_{j}+\sqrt{\varepsilon}}\frac{1}{2!}\ddot{H}(\tilde{y})(y-y_{j})^{2}dy\right]
\end{eqnarray}
where 
\begin{eqnarray}
y\equiv\varepsilon s,\ y_{j}\equiv\varepsilon jT_{0}^{2}=\varepsilon j\varepsilon^{-\frac{1}{2}}=\varepsilon^{\frac{1}{2}}j,\ \ \tilde{y}=\tilde{y}(y)\subset\left[y_{j},y_{j}+\sqrt{\varepsilon}\right].\label{eq:4.56}
\end{eqnarray}
The statement of the Lemma then follows by our previous boundless
assumption 
\begin{equation}
\left\Vert H\right\Vert <C
\end{equation}
\begin{equation}
\left\Vert \dot{H}\right\Vert <C
\end{equation}
\begin{equation}
\left\Vert \ddot{H}\right\Vert <C
\end{equation}
and 
\begin{equation}
\sup_{y}|y-y_{j}|\leq\sqrt{\varepsilon}.
\end{equation}
\hfill{}\qedsymbol

To continue the proof of theorem 4, assume $A=A(\varepsilon t),$
with $\sup_{0\leq s\leq1}\left\Vert \dot{A}\left(s\right)\right\Vert <C<\infty$.

To prove the uniformity in $\left(j,k\right)$ of the estimate in
\eqref{eq:440}, it is sufficient to prove the uniformity in $s\in\left[0,1\right]$
of the bound on \eqref{eq:4.53} with $T=1/\epsilon$.

Using the known ergodic (RAGE) estimate for time independent hamiltonians
due to \cite{key-3,key-4} $\left[\epsilon_{m},E-V_{ec}\right]$,
the needed bound for each \emph{fixed} s, will follow if the operator
$AP_{c}$ is compact and time independent.

To this end we write 
\begin{equation}
A\left(\epsilon t\right)P_{c}\left(\epsilon t\right)=A\left(s\right)P_{c}\left(s\right)+\int_{\epsilon^{-1}s}^{t}\frac{d}{dt''}A\left(\epsilon t''\right)P_{c}\left(\epsilon t''\right)dt''\label{eq:star_a}
\end{equation}
Using the assumed uniform boundedness of the derivatives $\dot{A},\dot{P}_{c},$
the bound \eqref{eq:4.62} - \eqref{eq:4.551} follows.

Now using the ergodic theorem (RAGE) for \eqref{eq:4.62}, given any
$\eta>0,$ it follows that for all $T\geq T_{\eta}\left(s\right)$,
the term \eqref{eq:4.62} is bounded by $\eta$. Hence $\left(4.53\right)$
follows. 

Given $\eta>0$, then for all $s\in[0,1]$, there exists time $T_{\eta}(s)$
s.t. $(t'=t-\varepsilon^{-1}s)$ 
\begin{eqnarray}
\|\frac{1}{T_{\eta}(s)}\int_{0}^{T_{\eta}(s)}e^{iH(s)t'}A(\epsilon t)P_{c}(\epsilon t)e^{-iH(s)t'}dt'\|\label{eq:4.53}\\
\leq\|\frac{1}{T_{\eta}(s)}\int_{0}^{T_{\eta}(s)}e^{iH(s)t'}A(s)P_{c}(H(s))e^{-iH(s)t'}dt'\|\label{eq:4.62}\\
+\|\frac{1}{T_{\eta}(s)}\int_{0}^{T_{\eta}(s)}T_{\eta}\left(s\right)\epsilon\left[\|A\|\|\dot{P}(s)\|+\|\dot{A}\|\|P(s)\|\right]dt'\|\label{eq:4.551}\\
\leq\eta+T_{\eta}\left(s\right)\varepsilon\sup_{s\in[0,1]}\left[\|A\|\|\dot{P}(s)\|+||\dot{A}\left(s\right)P_{c}\left(s\right)||\right].
\end{eqnarray}
The continuity of $H(s)$,$A\left(s\right)$ in $s$, implies that
if $|s-s'|$ is sufficiently small, then for all such $s,s'$, $\|H(s)-H(s')\|<\frac{1}{10T_{\eta}(s)}$
and $\left\Vert A\left(s\right)-A\left(s'\right)\right\Vert \leq\frac{\eta}{2}$.By
Lemma 3 the estimate (\ref{4.25}) holds for all such $H(s')$, with
the same $T_{\eta}(s)$. Note that the expression of (\ref{eq:4.62})
is equal to 
\begin{equation}
\|\frac{1}{T_{\eta}(s)}\int_{s\epsilon^{-1}}^{T_{\eta}\left(s\right)+\epsilon^{-1}s}e^{iH(s)t'}AP_{c}(H\left(s\right))e^{-iH(s)t'}dt'\|
\end{equation}
since change of variables $t'\rightarrow t'-\epsilon^{-1}s$ produces
factors $e^{t'iH\left(s\right)\left(\epsilon^{-1}s\right)}$ which
drop due to unitarity.

Therefore, by the compactness of {[}0,1{]}, it follows that there
is a finite subcover of such intervals, centered at $\{s_{i}\}_{i=1}^{K},$
$K<\infty$.

Hence, for $T_{\max}=\sup_{s_{i}}T_{\frac{\eta}{2}}(s_{i})$, we have
that for all $T_{\frac{\eta}{2}}\geq T_{max}$ 
\begin{eqnarray}
\|\frac{1}{T_{\eta}}\int_{s\varepsilon^{-1}}^{\varepsilon^{-1}s+T_{\eta}}e^{+iH(s)t'}A\left(\epsilon t\right)P_{c}(t)e^{-iH(s)t'}dt\|<\frac{\eta}{2}+C\epsilon\frac{T_{\eta}}{2}\label{eq:4.57}
\end{eqnarray}
 Choose $\varepsilon^{-1}=T_{\max}^{4}.$

By Lemma $\bar{H}$, for any $y_{i}\leq1/\epsilon$, $T_{0}\equiv\epsilon^{-1/2}$
we have : 
\begin{equation}
\|\bar{H}_{j}-H(y_{j})\|\leq O(\sqrt{\epsilon})\label{eq:455}
\end{equation}
and therefore by Lemma 3 we can apply the above estimate (\ref{eq:4.57})
for all $y_{j}\leq\frac{1}{\varepsilon}$ and $T_{\eta}\leq O(\varepsilon^{-1/2})=O\left(T_{max}^{2}\right)$,
choosing $s\epsilon^{-1}\equiv y_{j}$ in equation (\ref{eq:4.56}).

Collecting all of the above, we conclude that (choosing $T_{\frac{\eta}{2}}=\epsilon^{-1/2})$
\[
\left\Vert \sqrt{\epsilon}\int_{s\epsilon^{-1}}^{\epsilon^{-1}s+\epsilon^{-1/2}}e^{iH\left(s\right)t'}A\left(st\right)P_{c}\left(H\left(\epsilon t\right)\right)e^{-iH\left(s\right)t'}dt'\right\Vert <\eta
\]
uniformly in s.

By \eqref{eq:455} and Lemma 3, we can replace $H\left(s\right)$
by $\bar{H}_{j}$ for $y_{j}=\epsilon^{1/2}j=\epsilon s$ 

\textbf{Remark} as can be seen from the proof, the compact operator
$A$ can be adiabatically time dependent.

\noindent \textbf{Proof of the Gapless Adiabatic Theorem}

To complete the proof of the theorem we need to show the estimate
(\ref{eq:4.5} ). By the properties of the Kato dynamics, the integrand
can be written as (\ref{eq:4.9 fish}).

We need to prove \eqref{eq:4.6}. The integrand of \eqref{eq:4.6}
is given by the expression \ref{eq:4.9 fish}. The second term in
\ref{eq:4.9 fish} is controlled by the above ergodic theorem,

since by assumption $P_{0}$ is finite rank and hence compact. Moreover,
\begin{eqnarray}
\dot{P}_{0}(\varepsilon t)P_{0}(\varepsilon t)-\dot{P}_{0}(y_{j})P_{0}(y_{j})=O(\varepsilon T_{0})\ {\rm {if}\ |\varepsilon t-y_{j}|\leq T_{0}}
\end{eqnarray}
by arguments like \eqref{eq:star_a}
\begin{equation}
\partial_{t}\left(\dot{P}(\varepsilon t)P(\varepsilon t)\right)=O(\varepsilon)\left(\|\dot{P}\|+\|\ddot{P}\|\right).
\end{equation}

To deal with the first term in the expression \ref{eq:4.9 fish},
the term 
\begin{eqnarray}
U^{*}(t)\sum_{j\neq0}P_{j}(\varepsilon t)\dot{P}_{0}(\varepsilon t)P_{0}(\varepsilon t)U_{K}(t)\psi_{E_{0}}\label{eq:458}
\end{eqnarray}
(we can put the $U_{K}(t)$ dynamics back again) we approximate $U^{*}(t)P_{j}$
by Kato's dynamics again: 
\begin{eqnarray}
U^{*}(t)P_{j}(\varepsilon t)=U^{*}(t)U_{K_{j}}(t)U_{K_{j}}^{*}(t)P_{j}(\varepsilon t)
\end{eqnarray}
with $U_{K_{j}}(t)$ generated by $K_{j}=\lambda_{j}(\epsilon t')+i\epsilon\left[\dot{P}_{j}(\varepsilon t'),\ P_{j}(\varepsilon t')\right]$
so that 
\begin{eqnarray}
U_{K_{j}}(t)=\mathcal{T}e^{-i\int_{0}^{t}\lambda_{j}(\varepsilon t')-i\varepsilon\left[\dot{P}_{j}(\varepsilon t'),P_{j}(\varepsilon t')\right]dt'}
\end{eqnarray}
Then, the expression (\ref{eq:458}) takes the form 
\begin{eqnarray}
e^{i\int_{0}^{t}\lambda_{j}(s)ds}\left[U(t)^{*}U_{K_{j}}(t)\right]e^{-i\int_{0}^{t}\lambda_{j}\left(s\right)ds}U_{kj}^{*}P_{j}(\varepsilon t)\dot{P}_{0}(\varepsilon t)P_{0}(\varepsilon t)U_{K_{0}}(t)\psi_{E_{0}}
\end{eqnarray}
We integrate this last expression by parts, where the first integration
is of the composite integration of the phase factor. The rest of the
terms are (integrals of) derivatives of 
\begin{equation}
\left[U(t)^{*}U_{K_{j}}(t)\right],\ e^{-i\int_{0}^{t}\lambda_{i}\left(s\right)ds}U_{K_{j}}^{*}(t)P_{j}(\varepsilon t)\dot{P}_{0}(\varepsilon t)P_{0}(\varepsilon t)U_{K_{0}}(t)\psi_{0}
\end{equation}
and are all of order $\varepsilon$ (see \ref{eq:4.15} and \ref{eq:440}).
Without the loss of generality we assume $\lambda_{0}=0$. The composite
integration gives 
\begin{eqnarray}
\int_{0}^{t}dt'e^{i\int_{0}^{t'}\lambda_{j}(\varepsilon s)ds} & = & \frac{1}{\left(i\lambda_{j}(\varepsilon t)\right)}e^{i\int_{0}^{t}\lambda_{j}(\varepsilon s)ds}\\
 &  & -\frac{1}{\left(i\lambda_{j}(0)\right)}+\int_{0}^{t}\frac{dt'}{i\lambda_{j}^{2}}\varepsilon\dot{\lambda}_{j}(\varepsilon t')e^{i\int_{0}^{t'}\lambda_{j}(\varepsilon s)ds}.
\end{eqnarray}

Since we assume that there is a uniform gap between the eigenvalues,
$\left|\lambda_{j}\right|\geq\delta>0$, and therefore the above expression
is of order $\frac{1}{\delta^{2}}$.

We conclude that the contribution of \eqref{eq:458} to (\ref{eq:4.6})
is of order $\epsilon/\delta^{2}$ as is the case for Adiabatic theorem
with a gap.

\qedsymbol

Finally, we remark that if there is an eigenvalue crossing, finitely
many times, the argument of Kato applies to give another correction
of order $o_{\varepsilon}(1)$.

\section{Scattering Theory}

Scattering theory provides an important theoretical and computational
tool for analyzing adiabatic Hamiltonians. Recently, it became of
crucial importance to understand these processes in the study of soliton
(and other coherent) dynamics. Typically, linearizion around soliton,
and derivation of the modulation equation for its parameters lead
to matrix valued adiabatic Hamiltonians. Our goal in this section
is to show how the method of multi-time scale averaging, combined
with the Uniform Ergodic Theorem for adiabatic Hamiltonians lead to
scattering theory. In particular, we prove some classical propagation
estimates and asymptotic completeness.

The analysis of time dependent potential scattering is complicated
by the fact that energy is not conserved and cannot be localized.
Some parts of the solution run away to infinity in the energy, while
other parts concentrate on zero!

Propagation estimates were thus used, which included asymptotic localization
of the momentum/energy. In some special cases, where the interaction
is decaying fast in time enough one can prove the existence of asymptotic
distribution of energies \cite{key-5,key-2,key-3,key-4,key-9,Sig-Sofer2,Sig-Sofer3}
. Here, we show how the Ergodic theorem provides another route.

Consider the case where there are no bound states. The general case
will be considered elsewhere. $U(t)\psi$ has no bound states. Can
we prove scattering for $U(t)\psi$, where $U(t)$ is generated by
$H(\varepsilon t)$? We assume $H(\varepsilon t)=-\Delta+V(\varepsilon t,x)$,
where 
\begin{equation}
\sup_{t}\left|\left|\langle x\rangle^{\sigma}V(\varepsilon t,x)\right|\right|_{L^{\infty}}+\sup_{t}\left|\left|\langle x\rangle^{\sigma}\frac{\partial V(\tau,x)}{\partial\tau}\right|\right|_{L^{\infty}}<C_{0}.\label{eq:3.1}
\end{equation}
First we notice, that a small localized perturbation does not change
the scattering estimates in favorable cases.

\begin{prp} Assume $\psi=U^{(0)}(t)\psi$ solves the equation 
\begin{equation}
i\frac{\partial\psi}{\partial t}=H_{0}(t)\psi_{0},\qquad\psi_{0}=\psi(t=0)\label{eq:3.2}
\end{equation}
and such that 
\begin{equation}
\left|\left|\langle x\rangle^{-\sigma}\psi(t)\right|\right|_{L^{2}}\leq C_{\sigma,2}\langle t\rangle^{-\alpha},\qquad\alpha>1.\label{eq:3.3}
\end{equation}
Let $W_{\varepsilon}(x,t)$ be a localized function, bounded 
\begin{equation}
\sup_{t}\left|\left|\langle x\rangle^{+2\sigma}W_{\varepsilon}(x,t)\right|\right|_{L^{\infty}}<C_{0}\varepsilon.\label{eq:3.4}
\end{equation}
\label{eq:3.5} Then, if $\varepsilon$ is small enough, depending
on $C_{0}$, $C_{\sigma,2}$, we have 
\begin{equation}
\left|\left|\langle x\rangle^{-\sigma}U(t)\psi_{0}\right|\right|_{L^{2}}\leq C\left\langle t\right\rangle ^{-\alpha},
\end{equation}
where $U(t)\psi_{0}$ solves the equation 
\begin{equation}
i\frac{\partial}{\partial t}\left(U(t)\psi_{0}\right)=\left(H_{0}(t)+W_{\varepsilon}(x,t)\right)U(t)\psi_{0}.\label{eq:3.6}
\end{equation}
\end{prp} \begin{proof} 
\[
\psi(t)\equiv U(t)\psi_{0}=U^{(0)}(t)\psi_{0}-i\int_{0}^{t}U^{(0)}(t-s)W_{\varepsilon}(x,s)\psi(s)\mathrm{d}s.
\]
Then 
\begin{align}
\left|\left|\langle x\rangle^{-\sigma}\psi(t)\right|\right|_{L^{2}}\leq & \left|\left|\langle x\rangle^{-\sigma}U^{(0)}\psi(t)\right|\right|+\int_{0}^{t}\left|\left|\langle x\rangle^{-\sigma}U^{(0)}(t-s)\langle x\rangle^{-\sigma_{1}}\right|\right|_{L^{2}}\\
 & \times\left|\left|\langle x\rangle^{+\sigma_{1}}W_{\varepsilon}(x,s)\langle x\rangle^{+\sigma_{2}}\right|\right|_{L^{\infty}}\left|\left|\langle x\rangle^{-\sigma}\psi(s)\right|\right|_{L^{2}}\ \mathrm{d}s\\
\leq & C_{\sigma_{2}}\langle t\rangle^{-\alpha}+C_{\sigma,2}\int_{0}^{t}\langle t-s\rangle^{-\alpha}\langle s\rangle^{-\alpha}C_{\sigma,2}\epsilon\left|\left|\langle s\rangle^{\alpha}\langle x\rangle^{-\sigma}\psi(s)\right|\right|_{L^{2}}\ \mathrm{d}s
\end{align}
Hence, 
\begin{align}
\sup_{t\leq T}\ \langle t\rangle^{\alpha}\left|\left|\langle x\rangle^{-\sigma}\psi(t)\right|\right|_{L^{2}}\leq C_{\sigma,2}+\epsilon C_{0}C_{\sigma,2}\widetilde{C}\sup_{0\leq s\leq T}\left|\left|\langle s\rangle^{\alpha}\langle x\rangle^{-\sigma}\psi(s)\right|\right|_{L^{2}}
\end{align}
The result follows if $\epsilon C_{0}C_{\sigma,2}\widetilde{C}<1$.
\end{proof}

To Apply the above proposition to $U(t)\psi$ using the multiscale
time averaging, we need to verify that 
\begin{enumerate}
\item $U(t)\psi_{0}=U_{0}(t)U_{1}(t)\ldots U_{N}(t)\psi_{0}=V_{0,\varepsilon}(t)\psi_{0}$
where the generator of $V_{0,\varepsilon}$ is the generator of $U_{0}(t)$
plus a small localized perturbation for all $t$. 
 
\item That the scattering (local decay or similar) estimates hold for $U_{0}(t)\widetilde{\psi}$. 
\end{enumerate}
We want to apply the Ergodic theorem to the dynamics above. We will
assume $|V(x)|+|x\cdot\nabla V|\leq c\langle x\rangle^{-\sigma}$,
$\sigma>1$. The Ergodic theorem applies for any operator $C$ s.t.
\begin{eqnarray}
\sup_{j}\|\frac{1}{T}\int_{0}^{T}Ce^{-iH_{j}t}dt\|={o}_{T}(1),
\end{eqnarray}
so, it is sufficient to verify it for the time independent averaged
Hamiltonians.

If $C$ is compact, the result above is known as the RAGE theorem
(see e.g., CFKS). We want to apply it to a noncompact operator $C$
which is $O\left(\langle x\rangle^{-\sigma}\right)$ for some $\sigma>0$,
\textit{without} introducing high energy cutoff.

Local smoothing estimates hold for $H_{j}$: $(p\equiv-i\nabla_{x})$
\begin{eqnarray}
\int_{-\infty}^{\infty}\|\langle x\rangle^{-\frac{1}{2}-\varepsilon}|p|^{\frac{1}{2}}e^{-iH_{j}t}\psi\|_{L^{2}}^{2}dt<C\|\psi\|_{L^{2}}^{2}.
\end{eqnarray}
For each $T$, $\frac{1}{T}\int_{0}^{T}e^{-iH_{j}t}\langle x\rangle^{-\sigma'}e^{iH_{j}t}dt$
is a bounded self-adjoint operator. Hence, for $\sigma>1$, 
\begin{eqnarray}
\|\frac{1}{T}\int_{0}^{T}e^{iH_{j}t}\langle x\rangle^{-\sigma}e^{-iH_{j}t}dt\|\\
=\sup_{\|f\|_{L_{2}}=1}\frac{1}{T}\langle f,\int_{0}^{T}e^{iH_{j}t}\langle x\rangle^{-\sigma}e^{-iH_{j}t}fdt\rangle\\
=\sup_{\|f\|_{L_{2}}=1}\frac{1}{T}\int_{0}^{T}\|\langle x\rangle^{-\frac{\sigma}{2}}e^{-iH_{j}t}f\|^{2}dt<O\left(\frac{1}{T}\right).
\end{eqnarray}
Therefore, in this case we have the following Ergodic estimate 
\begin{eqnarray}
\|\frac{1}{T}\int_{0}^{T}U^{\ast}(\varepsilon t)\langle x\rangle^{-\sigma}U(\varepsilon t)dt\|=O\left(\frac{1}{T}\right).\\
\end{eqnarray}
For a general Hamiltonian, restricted to the continuous spectrum ,
we can expect that $O\left(\frac{1}{T}\right)$ be replaced by 
\begin{eqnarray}
O\left(T^{-\alpha}\right),\ \ \alpha\leq1.
\end{eqnarray}
We have the following immediate result:\\

\noindent \textbf{Theorem (Energy bound)}

For the above system the ``energy\textquotedbl{} can grow at most
by $O(\varepsilon)$, up to time of order $\frac{1}{\varepsilon}$.
Thus, the kinetic energy (the $H^{1}$ norm) can change by order 1.\\

\noindent \textbf{Proof}

\begin{eqnarray}
 &  & \|U^{\ast}(\varepsilon t)H(\varepsilon t)U(\varepsilon t)-U^{\ast}(0)H(0)U(0)\|\\
 &  & ~~~~~~~~~=\|\int_{0}^{t}U^{\ast}(\varepsilon t)\frac{\partial H(\varepsilon t)}{\partial t}U(\varepsilon t)dt\|\\
 &  & ~~~~~~~~~=\|\int_{0}^{t}U^{\ast}(\varepsilon t)\frac{\partial W(x,\varepsilon t)}{\partial t}U(\varepsilon t)dt\|\\
 &  & ~~~~~~~~~\lesssim\varepsilon\|\int_{0}^{t}U^{\ast}(\varepsilon t)\langle x\rangle^{-\sigma}U(\varepsilon t)dt\|\leq O(\varepsilon).
\end{eqnarray}
Since $H(\varepsilon t)=-\Delta+V(x,\varepsilon t)$, with $V$ bounded,
the result follows. \hfill{}\qedsymbol\\

\noindent \textbf{Asymptotic Completeness}

The proof of asymptotic completeness requires showing the following
strong limit : 
\begin{equation}
s-\lim e^{iH_{0}t}U(t)\psi
\end{equation}
exists for a dense set of $\psi$ as $t\rightarrow\pm\infty.$ By
Cook's method this is reduced to proving that 
\begin{equation}
\int^{\infty}\|W(x,\varepsilon t)U(t)\psi\|^{2}dt\le c<\infty.
\end{equation}
This last estimate follows directly from the above local decay estimate,
provided the Hamiltonian becomes time independent after time $1/\varepsilon.$

Moreover, if the potential is slowly and adiabatically becoming time
independent, that is if 
\begin{equation}
W(x,\varepsilon t)\equiv W_{0}(x,\varepsilon t)+(1+\varepsilon t)^{-a}W_{1}(x,\varepsilon t)
\end{equation}
with $W_{0}$ time dependent up to time $1/\varepsilon$, and $W_{1}$
nonzero for all times; $a>0.$ This last statement follows by the
following two simple observations: First , we can extend the proof
of local decay estimate to arbitrary time of size $M/\varepsilon$,
for $M$ large, and fixed, by changing $\varepsilon$ to a smaller
number, depending on $M.$ Then, after time $M/\varepsilon$, the
$W_{1}$ term is small, and absorbed as small perturbation of $H_{0}=-\Delta+W_{0}\left(x,\infty\right)$
as shown in the beginning of the section.

This last example is already of considerable interest, since such
Hamiltonians appear naturally in approximate nonlinear systems.

It is more difficult to get estimates for the case when the Hamiltonian
does not turn off. In this case the standard Adiabatic Theorem does
not apply, but one still would like to know the large time behavior
of the system. This is considered below.

We will now use {\em monotoic propagation observables} \cite{key-6}.
The operator we use is of the form 
\begin{equation}
\tanh(A/R).
\end{equation}
Here $A$ is the generator of dilation $A=\frac{1}{2}\left(x\cdot p+p\cdot x\right)$
and $R$ is a (large) constant. 

from this we derive the \cite{key-6} following general estimate:
\begin{theorem}{} Let the dynamics $U(t)$ be generated by the Hamiltonian
$-\Delta+W(x,t)$ and such that $W$ is a sufficiently regular function,
together with the derivatives $(x\cdot\nabla)^{n}W$, $0\le n\le N,N>3.$
We also assume that $W$ decays fast enough at infinity in x. Then,
the following propagation estimate holds:

\begin{align}
\frac{1}{R}\int_{0}^{T}\left\langle p\psi(t),F(\frac{|A|}{R}\le1)p\psi(t)\right\rangle dt\le\\
 & \frac{C}{R}\int_{0}^{T}\left\langle p\psi(t),<x>^{-\sigma}F(|x|\le X_{0})p\psi(t)\right\rangle dt+2\|\psi\|^{2},\nonumber 
\end{align}
where $\sigma$ is the decay rate of the potential, and for large
enough (depending W) $X_{0}$ 

\end{theorem}

\begin{theorem} Assume $H\left(\epsilon t\right)$ is given as above
and, $H\left(\epsilon t\right)$ is time dependent up to times $e^{\epsilon^{-\frac{1}{4}}}$\end{theorem}

Assume moreover that we have the following Ergodic type estimate:

\begin{equation}
\sup_{n}\int_{n/\sqrt{\epsilon}}^{(n+1)/\sqrt{\epsilon}}\left\langle p\psi(t),<x>^{-\sigma}p\psi(t)\right\rangle dt\le c.\label{eq:5.29}
\end{equation}
Then: 
\[
\int_{1}^{T}\left\langle p\psi(t),F(\frac{|A|}{R}\le1)p\psi(t)\right\rangle \frac{dt}{t}\le c\|\psi\|^{2}.
\]
\begin{proof} Then Applying Theorem 5 with $\left(R\equiv T/\ln^{2}T\right)$:
\begin{equation}
\frac{\ln^{2}T}{T}\int_{1}^{T}\left\langle p\psi(t),F(\frac{|A|}{T/\ln^{2}T}\le1)p\psi(t)\right\rangle dt\le2\|\psi\|^{2}+c\sqrt{\epsilon}\frac{T\ln^{2}T}{T},
\end{equation}

We therefore obtain, upon choosing $T\le e^{1/\varepsilon^{1/4}}$
that 
\begin{align}
\frac{\ln^{2}T}{T}\int_{1}^{T}\left\langle p\psi(t),F(\frac{|A|}{R}\le1)p\psi(t)\right\rangle dt\le\\
c\|\psi\|^{2},\nonumber 
\end{align}

for all
\[
R\le\frac{T}{\ln^{2}T}.
\]
Next, By integration by parts and the above estimate, we get

\begin{equation}
\int_{1}^{T}\left\langle p\psi(t),F(\frac{|A|}{R}\le1)p\psi(t)\right\rangle \frac{dt}{t}\le c\|\psi\|^{2}.\label{eq:5-3.1}
\end{equation}

Here $R,T$ as above and this estimate is uniform in $\epsilon\downarrow0$.\end{proof}

\begin{theorem}{}(Asymptotic completeness) \end{theorem}

AC holds for the above Hamiltonian (uniformly in $\epsilon\downarrow0$)
which depends on time up to $T\leq e^{\epsilon^{-1/4}}$ .

\begin{proof}

We need to prove that for a dense set of $\psi$: $s-limU_{0}^{*}\left(t\right)U\left(t\right)\psi$
exists, uniformly in $\epsilon$. 

We break the expression above as 
\[
U_{0}^{*}\left(t\right)F\left(\frac{\left|A\right|}{R}\leq1\right)U\left(t\right)\psi+U_{0}^{*}\left(t\right)F\left(\frac{\left|A\right|}{R}\geq1\right)U\left(t\right)\psi\equiv I_{1}+I_{2}
\]
To prove that $I_{2}$ has a strong limit, we use cook's method:
\[
\frac{d}{dt}U_{0}^{*}\left(t\right)FU\left(t\right)\psi=U_{0}^{*}\left(t\right)\left\{ i\left[H_{0},F\right]+FW\right\} U\left(t\right)\psi.
\]
The first term $i\left[H_{0},F\right]=\frac{1}{R}pF'\left(\frac{\left|A\right|}{R}\right)p+O\left(p^{2}/R^{2}\right)$.
Therefore the first term is integrable by theorem 6.

The second term is of the form 
\begin{eqnarray*}
e^{iH_{0}t}F\left(\frac{\left|A\right|}{R}>1\right)\left\langle x\right\rangle ^{-\sigma}U\left(t\right)\psi & =\\
e^{-iH_{0}t}F\left(\frac{\left|A\right|}{R}>1\right)F\left(\left|p\right|>kR\right)\left\langle x\right\rangle ^{-\sigma}\frac{1}{\left|p\right|}\left|p\right|U\left(t\right)\psi\\
+e^{iH_{0}t}F\left(\frac{\left|A\right|}{R}>1\right)F\left(\left|p\right|\leq kR\right)\left\langle x\right\rangle ^{\sigma}U\left(t\right)\psi\equiv J_{1}+J_{2}
\end{eqnarray*}
 The first term on the r.h.s is bounded by 
\[
\left\Vert J_{1}\right\Vert _{L^{2}}\leq c\sup_{\left\Vert f\right\Vert }\frac{1}{kR}\left\Vert \left\langle x\right\rangle ^{-\sigma/2}Fe^{-iH_{0}t}f\right\Vert \left\Vert \left\langle x\right\rangle ^{-\sigma/2}\left(1+\left|p\right|\right)U\left(t\right)\psi\right\Vert 
\]
\[
\left(\int_{0}^{T}\left\langle f,J_{1}\psi\right\rangle dt\right)\leq\frac{1}{kR}\int_{0}^{T}\left\Vert \left\langle x\right\rangle ^{-\sigma/2}Fe^{-iH_{0}t}f\right\Vert ^{2}dt+\frac{1}{kR}\int_{0}^{T}\left\Vert \left\langle x\right\rangle ^{-\sigma/2}\left(1+\left|p\right|\right)U\left(t\right)\psi\right\Vert ^{2}dt
\]
 The first term is bounded for all T, a property of the free flow
$H_{0}$.

The second term, for $T\leq e^{1/\epsilon^{1/4}}$ is bounded by(using
also \ref{eq:5.29}) 
\[
\frac{c}{kR}\sqrt{\epsilon}T.
\]
 we used that $F\left(\left|p\right|>kR\right)\left\langle x\right\rangle ^{-\sigma/2}\left|p\right|^{-1}=F\left|p\right|^{-1}\left\langle x\right\rangle ^{-\sigma/2}+F\left\langle x\right\rangle ^{-\sigma/2}\left[\left\langle x\right\rangle ^{\sigma/2},\frac{1}{\left|p\right|}\right]\left\langle x\right\rangle ^{-\sigma/2}=O\left(\frac{1}{kR}\right)\left\langle x\right\rangle ^{-\sigma/2}.$
The second term $J_{2}$ is bounded by $\left(k\ll1\right),\ \delta>0$
\begin{equation}
\sim e^{iH_{0}t}O\left(1\right)F\left(\left|x\right|>\delta R\right)\left\langle x\right\rangle ^{-\sigma}U\left(t\right)\psi=O\left(\delta^{-\sigma}R^{-\sigma}\right).
\end{equation}
 By choosing the range of integration $T\leq e^{1/\epsilon^{1/4}}$
and $R=T/\ln^{2}T,\ \sigma>1$ we get that
\[
J_{1}\leq\frac{c}{k}\sqrt{\epsilon}\ln^{2}T\leq\frac{c}{k}
\]
 
\[
J_{2}\leq\left(\delta^{-\sigma}T^{-\sigma}\ln^{2\text{\ensuremath{\sigma}}}T\right)T\leq CT^{-a},\ a>0.
\]
 Next we estimate $I_{1}.$

We claim that $I_{1}\rightarrow0$ as $t\rightarrow\infty$; uniformly
in $\epsilon\downarrow0.$ 

Estimate (\ref{eq:5-3.1}) implies that 
\begin{equation}
\frac{1}{T}\int_{0}^{T}\left\langle \psi\left(t\right),F\left(\left|p\right|\geq\delta\right)F\left(\frac{\left|A\right|}{R}\leq1\right)F\left(\left|p\right|\geq\delta\right)\psi\left(t\right)\right\rangle dt\leq\frac{C}{\delta^{2}\ln^{2}T}
\end{equation}
so, in particular $\left\langle \psi\left(t\right),F\left(\left|p\right|\geq\delta\right)F\left(\frac{\left|A\right|}{R}\leq1\right)F\left(\left|p\right|\geq\delta\right)\psi\left(t\right)\right\rangle $
is small for t large.

The estimate $I_{1}$:
\[
I_{1}\psi=e^{-iH_{0}t}F\left(\frac{\left|A\right|}{R}\leq1\right)F\left(p^{2}\leq\epsilon\right)U\left(t\right)\psi
\]
\[
\left\Vert I_{1}\psi\right\Vert \leq\left\Vert F\left(A\leq R\right)U\left(t\right)\psi\right\Vert =\left(U\left(t\right)\psi,F^{2}\left(A\leq R\right)U\left(t\right)\psi\right)^{1/2}
\]
 
\begin{eqnarray*}
\left(U\left(t\right)\psi,F^{2}\left(A\leq R\right)U\left(t\right)\psi\right)-\left(U\left(0\right)\psi,F^{2}\left(A\leq R\right)U\left(0\right)\psi\right) & =\\
\int_{0}^{t}\left\langle \psi\left(t'\right),i\left[H\left(t'\right),F^{2}\right]\psi\left(t'\right)\right\rangle dt'=\int_{0}^{t}\left\langle \psi\left(t'\right),\frac{1}{R}p\tilde{F}\left(A\right)p\psi\left(t'\right)\right\rangle dt'+\\
\text{\ensuremath{\int_{0}^{t}\left\langle \psi\left(t'\right),i\left[W\left(x,\epsilon t\right),F^{2})A\right]\psi\left(t'\right)dt'\right\rangle }}
\end{eqnarray*}
The first term is negative, as $\tilde{F}\sim F^{'2}\left(A\right).$
The second term the integrand is of the order $\left\langle \psi\left(t'\right),\frac{1}{R}\left\langle x\right\rangle ^{-\sigma}F'\left(A\right)\psi\left(t'\right)\right\rangle $.
If we integrate the second term, the up to time $T=e^{\epsilon^{-1/4}}$
and use the assumption (\ref{eq:5.29}) with $R\sim T/\ln^{2}T,$
the integral of the $W$ term is bounded by 
\[
C\frac{\ln^{2}T}{T}\sqrt{\epsilon}T=C\cdot1=C<\infty
\]
Since the first term is negative it is also uniformly bounded, by
the uniform boundness of the R.H.S .

So the convergence is independent of $\epsilon$.

Moreover if we let $t\rightarrow\infty$, then by assumption $H\left(\epsilon t\right)=H\left(\infty\right)$
for all $t>e^{\epsilon^{-1/4}}$. 

Therefore, the integrand can be extended to $t=\infty$.

It follows that 
\begin{align*}
U^{*}\left(t\right)F^{2}\left(A\leq R\right)U\left(t\right)\psi\xrightarrow{s}F^{\pm}\psi\\
\text{ as t}\rightarrow\pm\infty
\end{align*}
 Moreover, since for all $\left|t\right|>e^{\epsilon^{-1/4}}$ the
Hamiltonian is time independent and satisfies the usual scattering
and decay estimates, it follows that $F^{\pm}\equiv0$.\end{proof}

\section*{Acknowledgments}

This work was partly supported by the Israel Science Foundation (ISF
- 1028), by the US-Israel Binational Science Foundation (BSF -2010132),
by the USA National Science Foundation (NSF DMS 1201394)and by the
Shlomo Kaplansky academic chair.

\end{document}